\documentclass[11pt,reqno, final]{amsart}

\usepackage{amsmath,amssymb,amsthm,bbm,bm}
\usepackage{color}
\usepackage[colorlinks=true, allcolors=blue,backref=page]{hyperref}
\usepackage{mathrsfs}
\usepackage{mathtools}

\usepackage[noabbrev,capitalize,nameinlink]{cleveref}
\crefname{equation}{}{}
\usepackage[noadjust]{cite}

\usepackage{fullpage}

\usepackage{graphics}
\usepackage{pifont}
\usepackage[T1]{fontenc}
\usepackage{tikz}
\usetikzlibrary{arrows.meta}
\usepackage{environ}
\usepackage{framed}
\usepackage{url}
\usepackage[linesnumbered,ruled,vlined]{algorithm2e}
\usepackage[noend]{algpseudocode}
\usepackage[labelfont=bf]{caption}
\usepackage{cite}
\usepackage{framed}
\usepackage[framemethod=tikz]{mdframed}
\usepackage{appendix}
\usepackage{graphicx}
\usepackage[textsize=tiny]{todonotes}
\usepackage{tcolorbox}
\usepackage{enumerate}
\allowdisplaybreaks[1]
\usepackage{enumerate}

\crefname{algocf}{Algorithm}{Algorithms}

\crefname{equation}{}{} 
\AtBeginEnvironment{appendices}{\crefalias{section}{appendix}} 

\usepackage[color]{showkeys} 
\colorlet{refkey}{orange!20}
\colorlet{labelkey}{blue!30}

\crefname{algocf}{Algorithm}{Algorithms}

\numberwithin{equation}{section}
\newtheorem{theorem}{Theorem}[section]
\newtheorem{proposition}[theorem]{Proposition}
\newtheorem{lemma}[theorem]{Lemma}

\crefname{claim}{Claim}{Claims}

\newtheorem{corollary}[theorem]{Corollary}

\newtheorem*{question*}{Question}

\theoremstyle{definition}

\newtheorem*{definition*}{Definition}

\theoremstyle{remark}
\newtheorem*{remark}{Remark}

\newcommand{\snorm}[1]{\lVert#1\rVert}
\newcommand{\sang}[1]{\langle #1 \rangle}

\newcommand{\mb}{\mathbb}

\newcommand{\mbm}{\mathbbm}
\newcommand{\mc}{\mathcal}

\newcommand{\on}{\operatorname}

\newcommand{\assign}{\leftarrow}

\renewcommand{\epsilon}{\varepsilon}
\renewcommand{\tilde}{\widetilde}

\allowdisplaybreaks

\title{Spencer's Theorem in Nearly Input-Sparsity Time}

\author[A1]{Vishesh Jain}
\address{Department of Statistics, Stanford University, Stanford, CA 94305, USA}
\email{visheshj@stanford.edu}

\author[A2]{Ashwin Sah}
\author[A3]{Mehtaab Sawhney}
\address{Department of Mathematics, Massachusetts Institute of Technology, Cambridge, MA 02139, USA}
\email{\{asah,msawhney\}@mit.edu}

\thanks{Sah and Sawhney were supported by NSF Graduate Research Fellowship Program DGE-1745302. Sah was supported by the PD Soros Fellowship.}

\begin{document}

\maketitle
\begin{abstract}
A celebrated theorem of Spencer states that for every set system $S_1,\dots, S_m \subseteq [n]$, there is a coloring of the ground set with $\{\pm 1\}$ with discrepancy $O(\sqrt{n\log(m/n+2)})$. We provide an algorithm to find such a coloring in near input-sparsity time $\tilde{O}(n+\sum_{i=1}^{m}|S_i|)$. A key ingredient in our work, which may be of independent interest, is a novel width reduction technique for solving linear programs, not of covering/packing type, in near input-sparsity time using the multiplicative weights update method. 
\end{abstract}

\section{Introduction}

The celebrated `Six standard deviations suffice' theorem of Spencer \cite{Spe85} says that there is a universal constant $C > 0$ \footnote{In our notation, Spencer showed that $C < 6/\sqrt{\log{3}}$ suffices (when $m=n$), hence the name of the result.} such that given a set system $S_1,\dots, S_{m} \subseteq [n] := \{1,\dots, n\}$, there exists a bi-coloring $x := (x_1,\dots, x_n) \in \{\pm 1\}^{n}$ for which
\begin{equation}
\label{eq:spencer}
\max_{i \in [m]}\left|\sum_{j \in S_i}x_j\right| \leq C\sqrt{n\log(m/n + 2)}. 
\end{equation}
The strength of this result is most readily apparent when $m = n$, in which case the right hand side of \cref{eq:spencer} is $O(\sqrt{n})$ \footnote{Hadamard set systems show that this is optimal up to the implicit constant, see \cite[Theorem~13.4.1]{AS16}.} whereas a basic application of the probabilistic method only gives the weaker estimate $O(\sqrt{n\log{n}})$.

Spencer's proof of this result was non-algorithmic, based on the partial coloring technique of Beck \cite{Bec81}. A different (but still, non-algorithmic) proof, based on convex geometry, was obtained independently by Gluskin \cite{Glu88}. The problem of finding an $x \in \{\pm 1\}^{n}$, achieving the guarantee of \cref{eq:spencer}, in (probabilistic) polynomial time was first solved in a breakthrough work of Bansal \cite{Ban10}, by using a random walk guided by the solution to a semi-definite program (SDP) \footnote{Bansal's algorithm admits a derandomization. Another deterministic algorithm for Spencer's theorem was provided by Levy, Ramadas, and Rothvoss \cite{LRR17}.}. A much simpler random-walk based approach was later found by Lovett and Meka \cite{LM15}. Subsequently, Rothvoss \cite{Rot17} and Eldan and Singh \cite{ES18} devised randomized polynomial-time algorithms, based on convex geometry, which are also applicable to a generalization of Spencer's result due to Giannopoulos \cite{Gia97}. Without resorting to fast matrix multiplication (and restricting ourselves here to $m = \Theta(n)$ for ease of presentation), the fastest of these algorithms are those of Lovett--Meka and Eldan--Singh, both running in time $\tilde{O}(n^{3})$, where $\tilde{O}$ hides polylogarithmic factors in $n$. Allowing fast matrix multiplication, the running time of the algorithm of Eldan--Singh is dominated by the time to solve (to sufficient accuracy) a linear program with $n$ variables and $\Theta(n)$ constraints, for which the best-known bound is $\tilde{O}(n^{\omega})$ (where $\omega \approx 2.37$ is the \emph{current} best matrix-multiplication exponent \cite{AW21}) and $\tilde{O}(n^{2+1/18})$ even in the most optimistic case that the matrix-multiplication exponent is $2$ and the dual matrix-multiplication exponent is $1$ \cite{JSWZ20}. 

In the Workshop on Discrepancy Theory and Integer Programming in 2018 \cite{DadOpen}, devising discrepancy minimization algorithms running in near input-sparsity time was suggested as one of the main directions for future research, noting that such algorithms are not known for any of the major results in discrepancy theory, including Spencer's theorem, the Beck--Fiala theorem \cite{BF81}, or Banaszczyk's theorem \cite{Ban98}. In fact, all known algorithms for these problems (which produce a coloring with discrepancy at most a universal constant factor more than the best-known existential bounds) employ linear algebraic primitives such as solving a system of linear equations, which suffer from the matrix multiplication bottleneck. We note here that there are recent online algorithms for discrepancy minimization \cite{ALS21, LSS22} which do not use any such operations and run in input-sparsity time. However, these algorithms are only guaranteed to produce a coloring matching Banaszczyk's bound up to a polylogarithmic (as opposed to constant) factor. Moreover, in the setting of Spencer's theorem, Banaszczyk's bound itself only corresponds to a discrepancy of $O(\sqrt{n\log{n}})$, which is already attained with high probability by a uniformly random coloring.

In this work, we initiate the study of optimal discrepancy minimization algorithms running in near input-sparsity time, by providing such an algorithm for Spencer's theorem. 
\begin{theorem}\label{thm:main}
There exists an absolute constant $C_{\ref{thm:main}}>0$ and a randomized algorithm $\on{Coloring}$ such that the following holds. On input a matrix $A\in \mb{R}^{m\times n}$ such that $\snorm{A}_{1\to \infty}\le 1$, $\on{Coloring}(A)$ runs in time $\tilde{O}(\on{nnz}(A) + n)$ and with probability at least $1/2$, outputs a vector $v\in \{\pm 1\}^n$ such that \[\snorm{Av}_{\infty}\le C_{\ref{thm:main}}\sqrt{n\log(m/n + 2)}.\]
\end{theorem}
\begin{remark}
To see that this recovers \cref{eq:spencer}, we let the $i^{th}$ row of $A$ be the indicator vector of $S_{i}$ and note that $\snorm{A}_{1\to \infty} = \max_{i,j}|a_{i,j}| \leq 1$. 
\end{remark}

A pleasant feature of our algorithm (see \cref{sec:outline} for an overview) is that it relies not on a new approach to proving Spencer's theorem, but rather on solving the linear-program of Eldan--Singh in near-input sparsity time by leveraging certain key structural aspects, thereby raising the possibility that the linear algebraic primitives involved in other algorithms for discrepancy minimization may also be implemented much more efficiently. Moreover, during the course of our algorithm, we develop a novel method for solving a certain class of linear-programs (to polylogarithmic relative accuracy) in near-input sparsity time, which may be of independent interest as a rare instance of ``width reduction'' for linear-programs not of covering/packing type.  

\subsection{Proof outline}
\label{sec:outline}
We use the notation in \cref{eq:spencer}. For simplicity, we consider the case $m=n$ and $\sum_{i=1}^{n}|S_i| = \Theta(n^2)$ which already contains most of the ideas; later, we will sketch the modifications needed for the general case. For $C > 0$, we let
\[\Gamma_{C} := \left\{x \in [-1,1]^{n}: \max_{i \in [m]}\left|\sum_{j\in S_i}x_j\right| \leq C\sqrt{n}\right\}.\]
By a standard reduction (see \cref{thm:partial-coloring}), it suffices to devise an algorithm which runs in time $\tilde{O}(n^2)$ and finds $x \in \Gamma_{C}$ such that $x$ has at least $n/C$ coordinates which are $\pm 1$.

Our starting point for doing this is the aforementioned theorem of Eldan--Singh \cite{ES18} which shows that there is an absolute constant $C > 0$ such that with high probability over the choice of a random Gaussian vector $g \sim \mc{N}(0,1)^{\otimes n}$, 
\begin{align*}
    x^* := \arg\max_{x \in \Gamma_{C}}\langle g, x \rangle 
\end{align*}
has $n/C$ coordinates which are $\pm 1$. In particular, we can a find a point with the properties we want by solving the above linear program. In fact, a slight modification of the argument in \cite[Section~3]{Rot17} shows that it suffices to solve the linear program to within relative accuracy $1/\on{poly}(n)$ and then randomly round the approximate maximizer. However, solving the linear program to such small relative error forces us to use high-precision methods for solving linear programs, such as cutting-plane or interior-point methods, for which the matrix-multiplication bottleneck seems rather hard to circumvent. Instead, we rely on first-order optimization methods. 
\\  

\paragraph{\textbf{Stability of the linear program}:} In order to use first-order methods in time $\tilde{O}(n^2)$, we need to show that solving the linear program to within $1/\on{poly}(\log{n})$ relative accuracy (and then randomly rounding the approximate solution) suffices. More concretely, we show (\cref{prop:gauss-program}) that all points achieving at least a $(1-c_1/\log{n})$-factor of the optimum have $\Omega(n)$ coordinates with absolute value at least $1-c_2/\log{n}$; these coordinates can then be randomly rounded to have absolute value $1$, while only incurring additional discrepancy $O(\sqrt{n})$. Our proof of this shares some similarities to \cite{Rot17, ES18}, but ultimately relies on a different phenomenon: we use a supersaturation version of Spencer's theorem due to Spencer \cite{Spe85} as well as Gaussian concentration for Lipschitz functions to show (\cref{lem:large-C}) that the expected value of the linear program is $(\sqrt{2/\pi} - \delta_{C})n$, where $\delta_{C} \to 0$ as $n \to \infty$, which further enables us to show that the ``derivative'' of the map $C \mapsto \mb{E}[\max_{x \in \Gamma_{C}}\langle g, x\rangle]$ is sufficiently small, in a suitably averaged sense, for sufficiently large $C$ (\cref{lem:submultiplicative}). This gives us the desired conclusion since, for sufficiently large $C$, if there were a point within at least a $(1-c_1/\log{n})$-factor of $\mb{E}[\max_{x \in \Gamma_{C}}\langle g, x\rangle]$ with only sublinear coordinates with absolute value at least $1-c_2/\log{n}$, then it turns out we could find a point $x' \in \Gamma_{C(1+c_2/\log{n})}$ with $\langle g, x' \rangle - \mb{E}[\max_{x \in \Gamma_{C(1+c_2/\log{n})}}\langle g, x\rangle] = \tilde{\Omega}(n)$. However, by Gaussian concentration, this can only happen with probability at most $\exp(-\tilde\Omega(n))$ in the choice of $g$.\\

\paragraph{\textbf{The Multiplicative Weights Update (MWU) framework: }} Our task is now reduced to solving the linear program, which we rewrite in a more convenient form as a feasibility/search problem: for given $C > 0$, find $x$ such that  
\begin{align}
\label{eq:LP-original}
    x &\in [-1,1]^{n} \nonumber \\
    -\langle g/\sqrt{n}, x\rangle &\leq - \on{OPT}_{C} \nonumber \\
    Ax &\leq C \sqrt{n} \cdot \mbm{1},
\end{align}
where $\on{OPT}_{C} := \max_{x \in \Gamma_C}\langle g, x \rangle/\sqrt{n}$ \footnote{We do not know $\on{OPT}_C$ \emph{a priori}; however, by using a standard binary search procedure, we may assume access to a sufficiently good approximation, and we will ignore the distinction in this sketch.} and $A \in \mb{R}^{2n \times n}$ with rows $A_{2i} := \mbm{1}_{S_i}$, $A_{2i+1} := -\mbm{1}_{S_i}$. 
In fact, by the stability of the linear program, it suffices to find $x \in [-1,1]^{n}$ which satisfies the remaining inequalities up to an additive term of order $\sqrt{n}/\log{n}$ on the right hand side. By the usual multiplicative weights update method for solving linear programs (see, e.g., \cite{AHK12}), this can be done using iterations consisting of solving a \emph{single} linear inequality over the unit cube: find $x' \in [-1,1]^{n}$ s.t.
\begin{align}
\label{eq:iteration}
    -\rho_0 \langle g/\sqrt{n}, x' \rangle + \sum_{i=1}^{2n}\rho_i \langle A_i, x' \rangle \leq - \rho_0 \on{OPT}_{C} + (1-\rho_0)C\sqrt{n} + \sqrt{n}/\log{n},
\end{align}
where $\rho_0 , \rho_1,\dots, \rho_{2n}$ is a given probability distribution on $[2n+1]$ (corresponding to the ``weights of the experts'' in the MWU framework). Note that such an $x$, if it exists, can be found in time $O(n^2)$ using a greedy coordinate-by-coordinate strategy. 
Therefore, if $\tilde{O}(1)$ iterations of MWU were sufficient, then we would be done

The number of iterations required by MWU to give a solution with the required accuracy is
$\tilde{O}(w^{2}/n)$, where $w$ is the ``width'' of the procedure measuring the maximum violation of any constraint by the intermediate solutions $x'$. In our case, it is easy to see that
\begin{align*}
    w = O(\sqrt{n}) + \max_{x'} \langle g/\sqrt{n}, x' \rangle + \max_{x'} \max_{i \in [2n]}|\langle A_i, x'\rangle| = O(\sqrt{n}) + \max_{x'} \max_{i \in [2n]}|\langle A_i, x'\rangle|, 
\end{align*}
where $x'$ ranges over the solutions to the single linear inequality output by different iterations of MWU and where we have used that $\max_{x' \in [-1,1]^{n}}|\langle g/\sqrt{n}, x'\rangle| = O(\sqrt{n})$ as long as $\|g\|_{2} = O(\sqrt{n})$, which holds except with exponentially small probability over the choice of $g$. Therefore, if it were the case that
\[\max_{x'}\max_{i \in [2n]}|\langle A_i, x'\rangle| = \tilde{O}(\sqrt{n}),\]
then $\tilde{O}(w^{2}/n) = \tilde{O}(1)$, as required. Unfortunately, with the greedy coordinate-by-coordinate strategy for solving the single linear inequality on the cube, the width could potentially be of order $n$, so that $\tilde{O}(w^2/n) = \tilde{O}(n)$ and we get a total running time of $\tilde{O}(n^3)$, as for previously known algorithms.\\

\paragraph{\textbf{Width reduction}: }

To overcome this obstacle, we develop a novel ``width reduction'' technique (see the proof that \cref{prop:optimization} implies \cref{prop:oracle}), which guarantees that the width is $\tilde{O}(1)$ while blowing up the cost of each iteration by a factor of at most $\tilde{O}(1)$. In the optimization literature, width reduction techniques have been famously used to solve non-negative linear programs in near input-sparsity time (see, e.g., \cite[Table~1]{AO19}). However, our linear program is not of this form and our technique is completely different. In each iteration, instead of considering the simple linear inequality mentioned earlier, consider the following convex program: find $x' \in [-1,1]^{n}$ minimizing
\begin{align*}
    -\rho_0 \langle g/\sqrt{n}, x' \rangle + \sum_{i=1}^{2n}\rho_i \langle A_i, x' \rangle + \frac{1}{\on{poly}(\log   n)}\|Ax'\|_{\infty},
\end{align*}
possibly up to an additive error of $\sqrt{n}/(2\log n)$. This program has a few crucial properties:
\begin{itemize}
    \item By the feasibility of \cref{eq:LP-original}, it follows that the optimum value is at most the right hand side of \cref{eq:iteration}. Let $x'$ denote the (approximate) solution returned by the convex program. If $\|Ax'\|_{\infty} = \tilde{O}(\sqrt{n})$, then the point $x'$ is a solution to \cref{eq:iteration} with width $\tilde{O}(\sqrt{n})$, which suffices for our purpose. 
    
   \item By testing at $x' = 0$, we see that the optimum value is always non-positive. In particular, any minimizer $x' \in [-1,1]^{n}$ satisfies
    \[\sum_{i=1}^{2n}\rho_i \langle A_i, x' \rangle + \frac{1}{\on{poly}(\log n)}\|Ax'\|_{\infty} \le  \|g\|_{2}.\]
    Therefore, assuming $\|g\|_2 = O(\sqrt{n})$, if the minimizer $x'$ satisfies $\|Ax'\|_{\infty}\gg \sqrt{n}\on{poly}(\log n)$ then it must be that
    \[\sum_{i=1}^{2n}\rho_i \langle A_i, x' \rangle \ll -\sqrt{n}\on{poly}(\log  n),\]
    in which case $x'' = \tilde{O}(x'/\|Ax'\|_{\infty})$ is readily seen to satisfy \cref{eq:iteration} with width $\tilde{O}(1)$.
    
    \item The program can be written as a linear min-max program (or $\ell_{\infty}$--$\ell_{1}$ matrix game): 
    \[\min_{x' \in [-1,1]^{n}}\max_{y \in \Delta_{2n+1}} y_{0}\rho_0 \left(\langle g/\sqrt{n}, x' \rangle + \sum_{i=1}^{2n}\rho_i \langle A_i, x' \rangle\right) + \sum_{i=1}^{2n}y_{i}\langle A_i, x'\rangle/\on{poly}(\log n),\]
    where $\Delta_{2n+1} \subset \mb{R}^{2n+1}$ denotes the unit simplex consisting of probability distributions on $[2n+1]$, which one can try to solve using a saddle-point mirror descent scheme. At this point, it is natural to think that the correct geometry on $[-1,1]^{n}$ is given by the $\infty$-norm, so that one runs into the usual problem of the lack of an $\tilde{\Omega}(1)$-strongly convex mirror map on $[-1,1]^{n}$. However, we can take advantage of the fact that $\|A_i\|_{2} = O(\sqrt{n})$ in Spencer's problem (and that projecting onto $[-1,1]^{n}$ with respect to the Euclidean norm is easy) by viewing $[-1,1]^{n}$ as a subset of the $\ell_{2}$ ball of radius $\sqrt{n}$ and using the sublinear primal-dual framework of Clarkson, Hazan, and Woodruff \cite{CHW12} to show that the linear min-max program can indeed be solved to desired accuracy in time $\tilde{O}(n^2)$ (\cref{prop:opt-actual}). \\
\end{itemize}

\paragraph{\textbf{Running in near input-sparsity time}: }So far, we have assumed that $m = n$ and $\sum_{i=1}^{n}|S_i| = \tilde{O}(n^2)$. The same discussion extends to handle the case of general $m \leq n^{2}$ \footnote{If $m \geq n^2$, then a uniformly random coloring succeeds with high probability.} under the additional assumption that the set system is ``dense'' i.e. $\sum_{i=1}^{m}|S_i| = \tilde{O}(mn)$. To remove the density assumption, we require a few additional ingredients.  
\begin{itemize}
    \item First, we isolate variables which appear in at most $n/\on{poly}(\log n)$ sets and color them using the input-sparsity time online algorithm for Banaszcyzk's theorem due to Alweiss, Liu, and Sawhney \cite{ALS21} (\cref{app:deferred}). As mentioned earlier, this algorithm loses a polylogarithmic factor in the discrepancy guarantee, compared to Banaszczyk's theorem. However, since we are coloring only those variables which are in at most $n/\on{poly}(\log n)$ sets, such a loss is acceptable. 
    
    \item At this point, each of the $n_1$ remaining variables appears in at least $n/\on{poly}(\log n)$ sets and at most $m$ sets. Once again, we may assume that $m \leq n_1^{2}$, else a uniformly random coloring suffices. Let the total number of non-zero entries in this restricted incidence matrix be $N$, so that $n_1n/\on{poly}(\log n)\leq N \leq n_1m$. We isolate the at most $n_1/\log{n}$ variables which are present in at least $(\log n)\cdot N/n_1$ sets and color them uniformly at random. Note that a uniformly random coloring loses a factor of $\sqrt{\log{n_1}}$ compared to Spencer's bound, but this again acceptable, since we are only coloring at most $n_1/\log{n}$ variables.  
    
    \item Finally, all the remaining variables have the property of being present in at most $k = (\log n)\cdot N/n_1$ sets each and a careful implementation of \cite{CHW12} using a suitable (but fairly simple) data structure shows that the linear min-max program, restricted to such variables, can be solved in time $\tilde{O}(kn_1 + n_1^2) = \tilde{O}(N) = \tilde{O}(\sum_{i=1}^{m}|S_i|)$, as desired.     
\end{itemize}

\section{Preliminary reductions}
\label{sec:prelims}

In this section, we formally record a couple of preliminary reductions, which allow us to deduce \cref{thm:main} from a similar statement about appropriate `partial colorings' for input matrices $A$ with $\tilde{\Omega}(n^2)$ non-zero entries. 

First, by using an input-sparsity time online algorithm for the Koml\'os conjecture (with polylogarithmic losses) due to Alweiss, Liu, and the third author \cite{ALS21} (with an alternate proof by Liu and the last two authors \cite{LSS22}) on those variables which appear in $O(n/(\log n)^2)$ sets, it suffices to prove \cref{thm:main} with an extra additive $n^{2}$ in the running time.  

\begin{theorem}\label{thm:full-coloring}
There exists an absolute constant $C_{\ref{thm:full-coloring}}>0$ and a randomized algorithm $\on{Dense-Coloring}$ such that the following holds. On input a matrix $A\in \mb{R}^{m\times n}$ such that $\snorm{A}_{1\to \infty}\le 1$, $\on{Dense-Coloring}(A)$ runs in time $\tilde{O}(\on{nnz}(A) + n^2)$ and with probability at least $1/2$, outputs a vector $v\in \{\pm 1\}^n$ such that \[\snorm{Av}_{\infty}\le C_{\ref{thm:full-coloring}}\sqrt{n\log(m/n + 2)}.\]
\end{theorem}
The proof that \cref{thm:full-coloring} implies \cref{thm:main} is presented in \cref{app:deferred}.

Next, in order to prove \cref{thm:full-coloring}, it suffices to prove the following statement about colorings valued in $[-1,1]^{n}$ with linearly many coordinates colored by $\pm 1$. 

\begin{theorem}\label{thm:partial-coloring}
There exists an absolute constant $C_{\ref{thm:partial-coloring}} > 0$ and a randomized algorithm $\on{Partial-Coloring}$ such that the following holds. On input
\begin{itemize}
    \item a diagonal matrix $\Lambda\in \mb{R}^{n\times n}$ with $\snorm{\Lambda}_{1\to \infty}\le 1$, and
    \item a matrix $A \in \mb{R}^{m\times n}$ with $\snorm{A}_{1\to \infty}\le 1$,
\end{itemize}
$\on{Partial-Coloring}(A, \Lambda)$ runs in time $\tilde{O}(\on{nnz}(A) + n^2)$ and with probability at least $1/2$ outputs a vector ${v}\in [-1,1]^{n}$ such that 
\begin{itemize}
    \item $\sum_{i\in [n]}\mbm{1}_{|v_i| = 1}\ge C_{\ref{thm:partial-coloring}}^{-1} \cdot n$, and
    \item $\snorm{A\Lambda v}_{\infty}\le C_{\ref{thm:partial-coloring}} \sqrt{n\log(m/n + 2)}$.
\end{itemize}
\end{theorem}

The deduction of \cref{thm:full-coloring} from \cref{thm:partial-coloring} is, by now, standard, and essentially identical to that in Rothvoss \cite{Rot17}; we include the details in \cref{app:deferred} for the sake of completeness. 

\cref{thm:partial-coloring} follows from the following pair of propositions.

\begin{proposition}\label{prop:gauss-program}
There exist absolute constants $C_{\ref{prop:gauss-program}}, \eta_{\ref{prop:gauss-program}}>0$ such that the following holds. Given $A\in \mb{R}^{m\times n}$ such that $\snorm{A}_{1\to\infty}\le 1$,  independently sample $g\sim\mc{N}(0,1)^{\otimes n}$ and $C_{\on{Alg}}\sim\on{Unif}([C_{\ref{prop:gauss-program}}, 2C_{\ref{prop:gauss-program}}])$, and  let \[\Gamma_{A, C_{\on{Alg}}} := \{x\in \mb{R}^n:\snorm{Ax}_{\infty}\le C_{\on{Alg}}\sqrt{n\log(m/n+2)} \wedge \snorm{x}_{\infty}\le 1\}.\]
Then the following hold:
\begin{itemize}
    \item with probability at least $1-\exp(-\Omega(n))$, we have that 
    \[\snorm{g}_{2}\le 2\sqrt{n},\]
    \item with probability at least $4/5$, for any ${x}\in \Gamma_{A, C_{\on{Alg}}}$ such that 
    \[\sang{g,x}/\sqrt{n} \geq \sup_{y\in \Gamma_{A, C_{\on{Alg}}}}\sang{g,y}/\sqrt{n}  - \epsilon\eta_{\ref{prop:gauss-program}} \sqrt{n},\]
    we have that
    \[\sum_{j\in [n]}\mbm{1}_{|x_j|\ge 1-\epsilon}\ge \eta_{\ref{prop:gauss-program}}n,\]
    where $\epsilon = 1/\log n$.
\end{itemize}
\end{proposition}

\begin{proposition}\label{prop:solve-program}
For any $C\le \sqrt{\log n}/3$ and $n \geq 100$, there exists a randomized algorithm $\on{Solve}$ such that the following holds. On input 
\begin{itemize}
    \item a matrix $A\in \mb{R}^{m\times n}$ with $m \leq n^{2}$, $\snorm{A}_{1\to \infty}\le 1$, and $\max_{i\in [n]}|\on{supp}(Ae_i)|\le k$, and 
    \item a vector $v \in \mb{R}^{n}$ such that $\snorm{{v}}_{2}\le 2\sqrt{n}$, 
\end{itemize}
$\on{Solve}(A,v)$ runs in time $\tilde{O}(kn + n^2)$ and with probability at least $99/100$, outputs
\[z \in \Gamma_{A,C} := \{{x}:\snorm{{x}}_{\infty}\le 1 \wedge \snorm{A{x}}_{\infty}\le C\sqrt{n\log(m/n + 2)}\}\]
such that
\[\sang{{v}, {z}}/\sqrt{n} \geq \sup_{{y}\in \Gamma_{A, C}} \sang{{v},{y}}/\sqrt{n} - \epsilon \sqrt{n},\]
where $\epsilon = 1/(\log n)^2$.
\end{proposition}
\begin{remark}
Since ${0}\in \Gamma_{A,C} \subseteq \{y:\snorm{y}_2\le \sqrt{n}\}$, we have 
\begin{align}\label{eq:trivial-bound}
0\le \sup_{y\in \Gamma_{A, C}} \sang{v,y}/\sqrt{n}\le \sup_{\snorm{y}_2\le \sqrt{n}} \sang{v,y}/\sqrt{n}\le 2\sqrt{n}.
\end{align}
\end{remark}

We conclude this section by showing how to deduce \cref{thm:partial-coloring} from the above propositions.
\begin{proof}[{Proof of \cref{thm:partial-coloring}}]
By replacing $A$ by $A\Lambda$ (which can be computed in time $O(\on{nnz}(A) + n)$ and satisfies the same guarantee $\snorm{A\Lambda}_{1\to \infty}\le 1$), it suffices to consider the case when $\Lambda = I_{n}$. We may also assume that $m \leq n^{2}$ since if $m > n^{2}$, a direct application of Bernstein's concentration inequality and the union bound shows that a uniformly random assignment $v \sim \{\pm 1\}^{n}$ satisfies the conclusion of \cref{thm:partial-coloring} with high probability. 

Now, given $A \in \mb{R}^{m\times n}$ with $\|A\|_{1 \to \infty} \leq 1$ and $\Lambda = I_{n}$, define 
\[\mc{C}_{\on{Heavy}} := \{i: \on{supp}(Ae_i)\ge (\log n)\on{nnz}(A)/n\}.\]
We decompose $A$ into $A_1$ and $A_2$ with columns indexed by $\mc{C}_{\on{Heavy}}$ and $[n]\setminus \mc{C}_{\on{Heavy}}$ respectively. 

Note that $A_1\in \mb{R}^{m\times \mc{C}_{\on{Heavy}}}$, and by Markov's inequality, $|\mc{C}_{\on{Heavy}}|\le n/\log n$. Therefore, by Bernstein's inequality and a union bound, a uniformly random $v_1\in \{\pm 1\}^{\mc{C}_{\on{Heavy}}}$ satisfies $\snorm{A_1v_1}_{\infty}\le C\sqrt{n}$ with probability at least $9/10$, for a sufficiently large absolute constant $C > 0$.  

Thus it suffices to find a suitable coloring for $A_2\in \mb{R}^{m\times ([n]\setminus \mc{C}_{\on{Heavy}})}$. By \cref{prop:gauss-program,prop:solve-program}, in time $\tilde{O}(n\cdot ((\log n)\on{nnz}(A)/n) + n^2) = \tilde{O}(\on{nnz}(A) + n^2)$ and with probability at least $3/4$, we can find $v_{2}' \in [-1,1]^{[n] \setminus \mc{C}_{\on{Heavy}}}$ such that 
\begin{itemize}
\item $\|A_{2}v_{2}'\|_{\infty} \leq 2C_{\ref{prop:gauss-program}}\sqrt{n\log(2m/n + 2)}$, and 
\item $T: = \{j: |(v_2')_{j}| \geq 1-2/\log{n}\}$ satisfies $|T|\geq \eta_{\ref{prop:gauss-program}}n/2.$
\end{itemize}

Let $v_{2}$ be the random vector obtained from $v_2'$ by randomly rounding each coordinate in $T$ independently to $\{\pm 1\}$ in a manner such that $\mb{E}[(v_{2})_j] = (v_2')_{j}$ for all $j \in [T]$. Then, by Bernstein's inequality and the union bound, it follows that with high probability,
\begin{itemize}
\item $\|A_{2}v_{2}\|_{\infty} \leq 4C_{\ref{prop:gauss-program}}\sqrt{n\log(2m/n + 2)}$, and 
\item $T: = \{j: |(v_2)_{j}| =1 \}$ satisfies $|T|\geq \eta_{\ref{prop:gauss-program}}n/4,$
\end{itemize}
so that $v = (v_1, v_2)$ can be obtained in time $\tilde{O}(\on{nnz}(A) + n^2)$, and with probability at least $1/2$, satisfies the conclusion of \cref{thm:partial-coloring}. 
\end{proof}

\section{Stability of the linear program}\label{sec:gauss-stable}

The goal of this section is to prove \cref{prop:gauss-program}, which states that the linear program of Eldan and Singh \cite{ES18} is logarithmically stable, in the sense that all points with objective value within a $(1-c_{1}/\log{n})$-factor of the optimum have linearly many coordinates with absolute value at least $1-c_2/\log{n}$. 

As in the statement of the lemma, fix a matrix $A\in \mb{R}^{m\times n}$ such that $\snorm{A}_{1\to\infty}\le 1$, and for $C > 0$, define 
\[\Gamma_{A, C} := \{x\in \mb{R}^n:  \snorm{Ax}_{\infty}\le C\sqrt{n\log(m/n + 2)} \wedge \snorm{x}_{\infty}\le 1\}.\]
Furthermore, define 
\begin{equation}\label{eq:lin-prog}
\on{OPT}_{A}(C) := \mbm{E}_{g\sim \mc{N}(0,I_n)}\left[\max_{y\in \Gamma_{A,C}}\sang{g,y}\right]/\sqrt{n}.
\end{equation}

Our proof of \cref{prop:gauss-program} requires a few ingredients. First, we record a trivial upper bound on $\on{OPT}_{A}(C)$.

\begin{lemma}\label{lem:trivial-bound}
Fix $A\in \mb{R}^{m \times n}$ such that $\snorm{A}_{1\to\infty}\le 1$. For any $C > 0$,
\[\on{OPT}_{A}(C)\le \sqrt{2n/\pi}.\]
\end{lemma}
\begin{proof}
Since $\Gamma_{A,C} \subseteq [-1,1]^{n}$, we have 
\[\on{OPT}_{A}(C)\le \mbm{E}_{g\sim \mc{N}(0,I_n)}\left[\max_{\snorm{y}_{\infty}\le 1}\sang{g,y}\right]/\sqrt{n} = \sqrt{n}\mb{E}_{g\sim \mc{N}(0,1)}[|g|] = \sqrt{2n/\pi}. \qedhere\]
\end{proof}

Next, we record a consequence of Spencer's proof of Spencer's Theorem \cite{Spe85}. 

\begin{lemma}[{\cite[Theorem~12]{Spe85}}]\label{thm:many-points}
Fix $A\in \mb{R}^{m \times n}$ such that $\snorm{A}_{1\to\infty}\le 1$ and let $C > 0$. There exists a constant $\delta_C > 0$ (depending only on $C$ and independent of $A$) such that 
\[|\Gamma_{A, C} \cap \{\pm 1\}^n|\ge (2-\delta_C)^n,\]
where $\delta_C\to 0$ as $C\to \infty$.
\end{lemma}
\begin{remark}
In \cite[Theorem~12]{Spe85}, the corresponding result is only stated for $m = n$. A trivial modification of the proof however immediately gives the desired result. 
\end{remark}

We will also use concentration of Lipschitz functions with respect to the Gaussian measure. 

\begin{lemma}[see, e.g., {\cite[Theorem~2.2]{ES18}}]
\label{thm:gauss-concentration}
Let $f:\mb{R}^n\to\mb{R}$ be $L$-Lipschitz with respect to the Euclidean norm. Then
\[\mb{P}_{g\sim \mc{N}(0,1)^{\otimes n}}[|f(g)-\mb{E}[f(g)]|\ge t]\le 2\exp(-t^2/(2L^2)).\]
\end{lemma}

The previous two lemmas allow us to show that for sufficiently large $C$, the trivial upper bound in \cref{lem:trivial-bound} is close to sharp.
\begin{lemma}\label{lem:large-C}
Fix $A\in \mb{R}^{m \times n}$ such that $\snorm{A}_{1\to\infty}\le 1$ and let $C > 0$. There exists a constant $\delta'_C > 0$ (depending only on $C$ and independent of $A$) such that
\[\on{OPT}_{A}(C)\ge \sqrt{2n/\pi} - \delta'_C\sqrt{n},\]
where $\delta'_C\to 0$ as $C\to \infty$.
\end{lemma}
\begin{proof}
For $g \sim \mc{N}(0, I_n)$, note that $\on{sgn}(g)\in \{\pm 1\}^n$ is distributed uniformly on $\{\pm 1\}^{n}$. Furthermore if $\on{sgn}(g) \in \Gamma_{A,C}$, then 
\[\max_{y\in \Gamma_{A,C}}\sang{g,y} \ge \sang{g, \on{sgn}(g)}  = \snorm{g}_{1}.\]
In particular,
\begin{align*}
\mb{P}[\max_{y\in \Gamma_{A,C}}\sang{g,y} - \|g\|_{1}\ge 0]\ge \mb{P}[\on{sgn}(g) \in \Gamma_{A,C}] \geq \exp(-2\delta_Cn),
\end{align*}
where $\delta_C$ is as in \cref{thm:many-points}. 

Hence, since $g\mapsto \max_{y\in \Gamma_{A,C}}\sang{g,y}/\sqrt{n} - \snorm{g}_1/\sqrt{n}$ is $2$-Lipschitz with respect to the Euclidean norm, it follows from \cref{thm:gauss-concentration} that we must have 
\[\mb{E}[\max_{y\in \Gamma_{A,C}}\sang{g,y}/\sqrt{n} - \snorm{g}_1/\sqrt{n}]\ge - 8(\delta_C)^{1/2}\sqrt{n}\]
for sufficiently large $n$. The desired result now follows by the linearity of expectation.
\end{proof}

As a corollary, we deduce the following stability result for $\on{OPT}_{A}(C)$ with respect to $C$.

\begin{corollary}\label{lem:submultiplicative}
There exists an absolute constant $c_{\ref{lem:submultiplicative}} > 0$ and a non-increasing function $C_{\ref{lem:submultiplicative}} : [0,1] \to \mb{R}^{>0}$ for which the following holds. For any $0\le\epsilon\le c_{\ref{lem:submultiplicative}}$, $\delta \in [0,1]$, and $C\ge C_{\ref{lem:submultiplicative}}(\delta)$,
\[\mb{P}_{C'\sim \on{Unif}[C,2C]}[\on{Opt}_{A}(C'+C\cdot \epsilon)-\on{Opt}_{A}(C')\le \delta \cdot \epsilon\sqrt{n}]\ge 5/6.\]
\end{corollary}
\begin{remark}
It is an interesting problem in convex geometry to determine whether the Lipschitz constant of $C\mapsto \on{Opt}_{A}(\exp(C))$ is sufficiently small for all $C$ sufficiently large; our proof of \cref{lem:submultiplicative} shows that this is true in an appropriate averaged sense.
\end{remark}
\begin{proof}
For  $0\le i\le \lfloor 1/(9\epsilon) \rfloor$, consider $x_i := C + C\cdot i\cdot (9\epsilon)$. Note that for $C'\in [C+ 9iC \epsilon, C+(9i+8)C\epsilon]$, we have that $\on{OPT}_{A}(C'+\epsilon)-\on{OPT}_{A}(C')\le \on{OPT}_{A}(x_{i+1})-\on{OPT}_{A}(x_{i})$. Moreover,
\begin{align*}
\on{Opt}_{A}(C + \lfloor 1/(9\epsilon) \rfloor \cdot (9\epsilon)) - \on{Opt}_{A}(C)&=\sum_{i=0}^{\lfloor 1/(9\epsilon) \rfloor-1}\on{Opt}_{A}(x_{i+1}) - \on{Opt}_{A}(x_i)\\
&\ge \delta \cdot \epsilon \sqrt{n} \cdot \#\{i: \on{Opt}_{A}(x_{i+1}) - \on{Opt}_{A}(x_i)\ge \delta \cdot \epsilon \sqrt{n}\}.
\end{align*}
By \cref{lem:large-C}, the left hand side is at most $\delta_C'\sqrt{n}$, so that 
\[\#\{i: \on{Opt}_{A}(x_{i+1}) - \on{Opt}_{A}(x_i)\ge \delta \cdot \epsilon \sqrt{n}\}\le \delta'_{C}\delta^{-1}\epsilon^{-1}.\]
Therefore, 
\begin{align*}
\mb{P}_{C'\sim \on{Unif}[C,2C]}&[\on{Opt}_{A}(C'+C\cdot \epsilon)-\on{Opt}_{A}(C')]\le \delta \cdot \epsilon\sqrt{n}]\\
&\ge (8\epsilon) \cdot (\lfloor 1/(9\epsilon) \rfloor - \#\{i: \on{Opt}_{A}(x_{i+1}) - \on{Opt}_{A}(x_i)\ge \delta \cdot \epsilon \sqrt{n}\})\\
& \geq (8\epsilon)\cdot \lfloor 1/(9\epsilon) \rfloor - 8\delta'_{C}\delta^{-1}
\ge 5/6,
\end{align*}
provided $\epsilon$ is small and $C_{\ref{lem:submultiplicative}}(\delta)$ is chosen so that $\delta_{C_{\ref{lem:submultiplicative}}}'$ is sufficiently small compared to $\delta$.
\end{proof}

We are now in position to prove \cref{prop:gauss-program}.
\begin{proof}[Proof of \cref{prop:gauss-program}]
The first bullet point follows immediately from \cref{thm:gauss-concentration}, upon noting that  $\mb{E}[\snorm{g}_2] \leq (\mb{E}[\snorm{g}_2^2])^{1/2} = \sqrt{n}$ and that $g\mapsto \snorm{g}_2$ is a $1$-Lipschitz function of the Euclidean norm. 

We proceed to the proof of the second bullet point. Recall that $\epsilon = 1/\log n$ and let $\eta$ be a sufficiently small constant to be chosen later. Given $\eta$, we choose $C$ sufficiently large (according to  \cref{lem:submultiplicative}) so that $C_{\on{Alg}}$ satisfies $\on{Opt}_A(C_{\on{Alg}} + C_{\on{Alg}}\epsilon) - \on{Opt}_A(C_{\on{Alg}})\le \epsilon \eta\sqrt{n}$ with probability at least $5/6$. For the remainder of the proof, we fix $C, C_{\on{Alg}}$ satisfying this guarantee. Consider the random quantity 
\[\min_{|S| = (1-\eta)n}\max_{\substack{y\in \Gamma_{A,C_{\on{Alg}}}\\ \on{supp}(y)\in S}}\sang{g,y}/\sqrt{n}.\]
For a fixed $S$, we have that \[g \mapsto \max_{\substack{y\in \Gamma_{A,C_{\on{Alg}}}\\ \on{supp}(y)\in S}}\sang{g,y}/\sqrt{n}\]
is $1$-Lipschitz with respect to the Euclidean norm, and since $C$ is sufficiently large, we have by \cref{lem:large-C} that
\[\mb{E}\left[\max_{\substack{y\in \Gamma_{A,C_{\on{Alg}}}\\ \on{supp}(y)\in S}}\sang{g,y}/\sqrt{n}\right]\ge \sqrt{n}/2.\]
Therefore, by \cref{thm:gauss-concentration} and a union bound over the $\binom{n}{\eta n}\le 2^{n/100}$ choices for $S$, it follows that with probability at least $1 - \exp(-\Omega(n))$ over the choice of $g$,
\begin{align}\label{eq:adjuster}
\min_{|S| = (1-\eta)n}\max_{\substack{y\in \Gamma_{A,C_{\on{Alg}}}\\ \on{supp}(y)\in S}}\sang{g,y}\ge \sqrt{n}/4.  
\end{align}
Moreover, by our choice of $C_{\on{Alg}}$ and \cref{thm:gauss-concentration}, we have that
\begin{align}
\label{eq:stable-max}
\mb{P}\left[\max_{y'\in \Gamma_{A,C_{\on{Alg}}+ C_{\on{Alg}}\epsilon}} \sang{{g},{y}'}/\sqrt{n} - \max_{y\in \Gamma_{A,C_{\on{Alg}}}}\sang{g,y}/\sqrt{n} < \epsilon \sqrt{n}/8\right]\ge 1 - \exp(-\Omega(\epsilon^2n)).
\end{align}

To conclude the proof, we show that on the events appearing in \cref{eq:adjuster} and \cref{eq:stable-max} (which simultaneously hold with probability $1 - 2\exp(-\Omega(\epsilon^{2}n))$, and for $\eta \leq 1/16$, there cannot exist a vector ${x} \in \Gamma_{A, C_{\on{Alg}}}$ such that
\[\sang{{g},{x}}/\sqrt{n} - \sup_{{y}\in \Gamma_{A, C_{\on{Alg}}}}\sang{{g},{y}}/\sqrt{n} \ge - \epsilon\eta \sqrt{n} \quad \text{and} \quad \sum_{j\in [n]}\mbm{1}_{|x_j|\ge 1-\epsilon}\le \eta n.\]
Indeed, suppose for contradiction that such a vector $x$ exists. Let $S$ be the set of coordinates where ${x}$ has magnitude less than $1-\epsilon$. By \cref{eq:adjuster}, there exists ${z}\in \Gamma_{A,C_{\on{Alg}}}$ such that \[\sang{{g},{z}}/\sqrt{n}\ge \sqrt{n}/4.\]
Consider the vector ${z}' = {x} + \epsilon \cdot {z} \in  \Gamma_{A, C_{\on{Alg}} + C_{\on{Alg}}\epsilon}$. We have that  
\[\sang{{g},{z}'}/\sqrt{n} - \max_{y\in \Gamma_{A,C_{\on{Alg}}}}\sang{g,y}/\sqrt{n} \ge \epsilon \sqrt{n}/4-\epsilon \eta \sqrt{n}\ge \epsilon \sqrt{n}/8,\]
and hence,
\[\max_{y'\in \Gamma_{A,C_{\on{Alg}}+ C_{\on{Alg}}\epsilon}} \sang{{g},{y}'}/\sqrt{n} - \max_{y\in \Gamma_{A,C_{\on{Alg}}}}\sang{g,y}/\sqrt{n} \ge \epsilon \sqrt{n}/8,\]
which cannot happen on the event appearing in \cref{eq:stable-max}.
\end{proof}

\section{Reduction to a Minimax Problem}

The remainder of this paper is devoted to establishing \cref{prop:solve-program}. Since we only need to solve the linear program in \cref{prop:solve-program} to $(1-1/\on{poly}(\log n))$-relative error, it is natural to consider a first-order method for solving linear programs, such as using the multiplicative weights update (MWU) method (see \cite{AHK12} for an excellent introduction). Unfortunately, one immediately runs into the issue that the so-called width of the linear program can potentially by $\Theta(n)$, so that the MWU-based solver takes time $\tilde{O}(n^{3})$. To overcome this obstacle, we introduce a novel method of `width reduction', which takes advantage of the structure of our linear program.

In the next section, we will show how to implement the following key subroutine which, combined with the standard MWU approach, proves \cref{prop:solve-program}.
\begin{proposition}\label{prop:oracle}
For any $C \leq \sqrt{\log{n}}/3$ and $n\geq 100$, there exists a randomized algorithm $\on{Regularized-Solve}$ for which the following holds. On input:
\begin{itemize}
    \item a matrix $A\in \mb{R}^{m\times n}$ with $m \leq n^{2}$, $\snorm{A}_{1\to \infty}\le 1$, and $\max_{i\in [n]}|\on{supp}(Ae_i)|\le k$, 
    \item $\rho_0\in \mb{R}^{\ge 0}$, $\rho_{+}, \rho_{-}\in (\mb{R}^{\ge 0})^{m}$ with $\rho_0 + \snorm{\rho_{+}}_1 + \snorm{\rho_{-}}_1= 1$, 
    \item ${v} \in \mb{R}^{n}$ with $\snorm{{v}}_2 \leq 2\sqrt{n}$, 
    \item $\Lambda \in [0, 2\sqrt{n}]$,
\end{itemize}
$\on{Regularized-Solve}(A,\rho_0, \rho_{+},\rho_{-}, {v}, \Lambda)$ runs in time $\tilde{O}(kn + n^2)$, and with probability at least $1-1/\sqrt{n}$, outputs either: 
\begin{itemize} 
    \item a certificate that the program \[\on{sup}_{x\in \Gamma_{A,C}}\sang{x,v}\ge \Lambda\sqrt{n}\] is not feasible, where
    \[\Gamma_{A,C} := \{x\in \mb{R}^n : \snorm{x}_{\infty}\le 1 \wedge \snorm{Ax}_{\infty}\le C\sqrt{n\log(m/n+2)}\}, \quad \text{or}\]
    \item a point $x\in [-1,1]^{n}$ satisfying $\snorm{Ax}_{\infty}\le 10\sqrt{n}(\log n)^4$ and 
    \begin{align*}
    -\rho_0 {v}^Tx/\sqrt{n} &+ \rho_{+}^{T}Ax - \rho_{-}^{T}Ax\\
    &\le -\rho_0\Lambda  + (C\sqrt{n\log(m/n+2)})(\snorm{\rho_{+}}_{1}+ \snorm{\rho_{-}}_{1}) + \sqrt{n}/(\log n)^{3}
    \end{align*}
\end{itemize}
\end{proposition}

\begin{remark}
Note that it is easy to compute $\arg\min_{x\in [-1,1]^{n}}-\rho_0 v^{T}x + \rho_+^{T} Ax - \rho_-^{T} Ax$, which coincides with $-\on{sgn}(-\rho_0 v^{T}/\sqrt{n} + \rho_{+}^{T}A - \rho_{-}^{T}A)$, in time $O(\on{nnz}(A) + n)$. The content of this proposition is that we can efficiently find an approximate optimizer $x$ with $\|Ax\|_{\infty} = \tilde{O}(\sqrt{n})$, which is best possible up to logarithmic factors. 
\end{remark}

\begin{proof}[{Proof of \cref{prop:solve-program}} given \cref{prop:oracle}]

We use the notation in \cite[Section~3.3]{AHK12}. For convenience of notation, let $C_{m,n} := C\sqrt{n\log(m/n+2)}$. In order to prove \cref{prop:solve-program}, it suffices to solve the feasibility program
\begin{align*}
    \exists? x \in [-1,1]^{n}: \quad  \frac{v^Tx}{\sqrt{n}}\ge \Lambda; \quad Ax\ge  -C_{m,n}\mbm{1}; \quad -Ax\ge -C_{m,n}\mbm{1},
\end{align*}
where $\Lambda := \sup_{y \in \Gamma_{A,C}}\sang{v,y}/\sqrt{n}$ \footnote{The value of $\Lambda$ is not known to us. However, by combining the discussion here with a standard binary search routine, we can approximate $\Lambda$ to within additive error $\sqrt{n}/(\log{n})^{4}$ with probability $1-o_n(1)$. The binary search procedure only blows up the overall running time by a factor of $O(\log\log{n})$, since $\Lambda \in [0,2\sqrt{n}]$ by \cref{eq:trivial-bound}.}, up to an additive error of $\epsilon = \sqrt{n}/(\log{n})^{3}$ on the right hand side for each of the three inequalities. Indeed, given such a point $x$, it is readily seen that $z := (1-1/(\log n)^{3})\cdot x$ satisfies the conclusion of \cref{prop:solve-program}. 

For this feasibility program, we begin by noting that \cref{prop:oracle} provides an $(\ell, \rho)$-bounded $\textsc{Oracle}$, in the sense of \cite[Definition~3.2]{AHK12}, with $\ell = \rho = 20\sqrt{n}(\log{n})^{4}$, and with probability at least $1-1/\sqrt{n}$. Indeed, the point $x \in [-1,1]^{n}$ output by \cref{prop:oracle} satisfies
\[\max_{i\in [m]}|A_{i}x - C_{m,n}| \leq \|Ax\|_{\infty} + C_{m,n} \leq 20\sqrt{n}(\log n)^{4},\]
where the final inequality holds since $m \leq n^{2}$ and $C \leq \sqrt{\log{n}}$. Also, 
\[|\sang{v,x}/\sqrt{n} - \Lambda| \leq \Lambda + \|v\|_{2}\|x\|_{2}/\sqrt{n} \leq 4\sqrt{n}.\]

Therefore, by \cite[Theorem~3.3]{AHK12}, the MWU algorithm makes $O(\ell \rho (\log m)/\epsilon^{2}) = \tilde{O}(1)$ calls to the \textsc{Oracle} in \cref{prop:oracle}, with an additional processing time of $O(m)$ per call, and provided that all the calls to the oracle succeed, either finds $x \in [-1,1]^{n}$ solving the feasibility program to within an additive error of $\epsilon$ on the right hand side for each of the three inequalities, or correctly concludes that the system is infeasible. Since each call to the oracle succeeds with probability at least $1-1/\sqrt{n}$, it follows by a union bound that all calls succeed with probability at least $1-\tilde{O}(1)/\sqrt{n}$. Moreover, the total running time is $\tilde{O}(m) + \tilde{O}(kn + n^{2}) = \tilde{O}(kn + n^{2})$, where we have used that $m \leq n^{2}$.
\end{proof}

In the next section, we will show how to implement the algorithm $\on{Regularized-Solve}$. Our construction crucially relies on the following reduction to solving a minimax program, whose solutions can be transformed into low width solutions for the MWU algorithm by using the structure of our linear program. 

\begin{proposition}\label{prop:optimization}
There exists a randomized algorithm $\on{Optimize}$ for which the following holds. On input
\begin{itemize}
    \item a matrix $A\in \mb{R}^{m\times n}$ with $m\leq n^{2}$, $\snorm{A}_{1\to \infty}\le 1$, and $\max_{i\in [n]}|\on{supp}(Ae_i)|\le k$,
    \item $\rho_0\in \mb{R}^{\ge 0}$, $\rho_{+}, \rho_{-}\in (\mb{R}^{\ge 0})^{m}$ with $\rho_0 + \snorm{\rho_{+}}_1 + \snorm{\rho_{-}}_1= 1$,
    \item ${v} \in \mb{R}^{n}$ with $\snorm{{v}}_2\in [0, 2\sqrt{n}]$,
    \item $\delta\in (0,1)$,
\end{itemize}
$\on{Optimize}(A,\rho_0, \rho_{+},\rho_{-}, {v}, \delta)$ runs in time $\tilde{O}(kn + n^2)$, and with probability at least $1-1/n$, outputs a point $x'\in [-1,1]^{n}$ such that 
\begin{align*}
&- \rho_0 v^{T}x'/\sqrt{n} + \rho_{+}^{T}Ax' -  \rho_{-}^{T}Ax'+ \delta \snorm{Ax'}_{\infty}\\
&\le \min_{x \in [-1,1]^{n}} - \rho_0 v^{T}x/\sqrt{n} + \rho_{+}^{T}Ax -  \rho_{-}^{T}Ax+ \delta \snorm{Ax}_{\infty} + \sqrt{n}/(\log n)^{4}.
\end{align*}
\end{proposition}

We conclude this section by showing how to transform $x'$ into a point satisfying the conclusion of \cref{prop:oracle}.
\begin{proof}[Proof of \cref{prop:oracle} given \cref{prop:optimization}]
Let $\delta = 1/(\log n)^{4}$ and $C_{m,n} := C\sqrt{n\log(m/n+2)}$. For $\rho = (\rho_0, \rho_+, \rho_-)$, define
\begin{align}\label{eqn:penalty}
\on{OPT}_{\rho} := \min_{x \in [-1,1]^{n}} - \rho_0 v^{T}x/\sqrt{n} + \rho_{+}^{T}Ax -  \rho_{-}^{T}Ax+ \delta \snorm{Ax}_{\infty}.
\end{align}
Note that $\on{OPT}_{\rho} \leq 0$, since $0 \in [-1,1]^{n}$. 
Under the assumption that the linear program
\[\sup_{x \in \Gamma_{A,C}} \sang{x,v} \geq \Lambda\]
is feasible, we have that 
\begin{align}
\label{eq:if-feasible}
\on{OPT}_{\rho}&\le -\rho_0\Lambda + (\snorm{\rho_{+}}_1 + \snorm{\rho_{-}}_1)C_{m,n} + \delta C_{m,n} \nonumber\\
&\le -\rho_0\Lambda + (\snorm{\rho_{+}}_1 + \snorm{\rho_{-}}_1) C_{m,n} + 2C\sqrt{n}/(\log n)^{7/2},
\end{align}
where we have used the value of $\delta$ and the assumption $m\leq n^{2}$.

Moreover, it follows from \cref{prop:optimization} that in time $\tilde{O}(nk + n^2)$ and with probability at least $1-1/n$, we can produce a vector $x \in [-1,1]^{n}$ such that
\[- \rho_0 v^{T}x/\sqrt{n} +  \rho_{+}^{T}Ax -  \rho_{-}^{T}Ax + \delta \snorm{Ax}_{\infty} \le \on{OPT}_{\rho} + \sqrt{n}/(\log n)^{4}.\]

To finish the proof, we show how to convert this $x$ into a satisfying assignment for the second bullet point of \cref{prop:solve-program} (if the linear program is feasible) or to certify infeasibility otherwise. Let $\tau := \delta \|Ax\|_{\infty}$, which can be computed in time $O(\on{nnz}(A) + \max(m,n))$. We have the following cases:

\textbf{Case I: } $\tau \geq 10\sqrt{n}$. We have
    \[- \rho_0 v^{T}x/\sqrt{n} +  \rho_{+}^{T}Ax -  \rho_{-}^{T}Ax \le \on{OPT}_{\rho} + \sqrt{n}/(\log{n})^{4} - \tau \leq  \sqrt{n}-\tau.\]
    Now $y := x\cdot 10\sqrt{n}/\tau$ satisfies the conclusion of \cref{prop:oracle} since $y \in [-1,1]^{n}$, $\|Ay\|_{\infty} \leq 10\sqrt{n}\delta^{-1} = 10\sqrt{n}(\log{n})^{4}$, and
    \begin{align*}
        - \rho_0 v^{T}y/\sqrt{n} +  \rho_{+}^{T}Ay -  \rho_{-}^{T}Ay &\le (\sqrt{n}-\tau) \cdot 10\sqrt{n}/\tau \le -5\sqrt{n}\\
        &\le -\Lambda \le - \rho_0 \Lambda \\
        &\le -\rho_0 \Lambda + C_{m,n}(\|\rho_+\|_1 + \|\rho_-\|_1).
    \end{align*}

\textbf{Case II: }$\tau < 10\sqrt{n}$. Then, $\|Ax\|_{\infty} \leq 10\sqrt{n}(\log{n})^4$ and 
    \begin{align*}
        - \rho_0 v^{T}x/\sqrt{n} +  \rho_{+}^{T}Ax -  \rho_{-}^{T}Ax \le \on{OPT_{\rho}} + \sqrt{n}/(\log n)^{4}.
    \end{align*}
    We compute the left hand side of the above inequality in time $O(\on{nnz}(A) + m)$. If it is at most 
    \[-\rho_0 \Lambda + (\|\rho_+\|_{1} + \|\rho_-\|_{1})C_{m,n} + 2C\sqrt{n}/(\log{n})^{7/2} + \sqrt{n}/(\log{n})^{4},\]
    then $x$ satisfies the conclusion of \cref{prop:oracle}. 
     If not, then it must be the case that
    \[\on{OPT}_{\rho} > -\rho_0 \Lambda + (\|\rho_+\|_{1} + \|\rho_-\|_{1})C_{m,n} + 2C\sqrt{n}/(\log{n})^{7/2},\]
    which certifies by \cref{eq:if-feasible} that the original linear program is not feasible.  \qedhere


\end{proof}

\section{Solving the Minimax Program via Sublinear Primal-Dual Algorithm}\label{sec:opt}

Finally, we prove \cref{prop:optimization} in the following more general form. 
\begin{proposition}\label{prop:opt-actual}
There is a randomized algorithm $\on{Optimize}$ for which the following holds. On input
\begin{itemize}
    \item $v_1,\ldots, v_m\in \mb{R}^{n}$ with $\snorm{v_i}_{2}\le 1/2$ and $\max_{i \in [n]}\#\{j:\sang{v_j,e_i}\neq 0\}\le k$,
    \item $v' \in \mb{R}^{n}$ with $\snorm{v'}_2\leq 1/2$,
    \item $\epsilon\in (0,1)$,
\end{itemize}
$\on{Optimize}(v', v_1,\ldots, v_m, \epsilon)$ runs in time $\tilde{O}((n+k)\epsilon^{-2} + m)$ \footnote{We make the standard data-structure assumption that for any $i \in [n]$, the set $\{j: \sang{v_j, e_i} \neq 0\}$, which has size at most $k$ by assumption, can be determined in time $\tilde{O}(k)$.}, and with probability at least $1-1/n$, outputs a point $x'\in [-1/\sqrt{n},1/\sqrt{n}]^{n}$ such that 
\begin{align*}
& \max_{j\in [m]} (v' + v_j)^{T}x' \le \epsilon + \min_{x \in [-1/\sqrt{n},1/\sqrt{n}]^{n}} \max_{j \in [m]} (v' + v_j)^{T}x. 
\end{align*}
\end{proposition}

We claim that \cref{prop:opt-actual} implies \cref{prop:optimization}. This is due to the following set of observations. To disambiguate notation, let $A \in \mb{R}^{m'\times n}$ be the matrix appearing in \cref{prop:optimization}.
\begin{itemize}
    \item We set $m = 2m'$. For $i \in [m']$, we let $v_{i} := \delta \cdot e_{i}^{T}A/(2\sqrt{n})$ and $v_{m'+i} = -\delta \cdot e_{i}^{T}A/(2\sqrt{n})$. The assumptions on $v_1,\dots,v_m$ are satisfied by the first bullet point of \cref{prop:optimization} and since $\delta \in (0,1)$. Moreover, for any $x \in \mb{R}^{n}$,
    \[\frac{\delta}{2\sqrt{n}} \|Ax\|_{\infty} = \max_{i \in [m]}v_{i}^{T}x.\]
    \item We let $v' := (-\rho_0 v^{T}/\sqrt{n} + \rho_+^{T} A - \rho_-^T A)/2\sqrt{n}$. Note that this vector can be computed in time $O(\on{nnz}(A) + \max\{m,n\}) = O(\on{nnz}(A) + n^2) = O(nk + n^2)$. The norm assumption on $v'$ holds since
    \[\snorm{-\rho_0{v}^T/\sqrt{n} + \rho_{+}^{T}A-\rho_{-}^{T}A}_2 \leq \max\{\snorm{v}_2/\sqrt{n}, \snorm{A}_{2\to\infty}\}\le \sqrt{n},\]
    where the final inequality uses the first and third bullet points of \cref{prop:oracle}. 
    \item Finally, setting $\epsilon = 1/(2\sqrt{n}(\log{n})^{4})$, it is immediately seen that if $x'$ satisfies the conclusion of \cref{prop:opt-actual}, then $x = 2\sqrt{n}\cdot x'$ 
    satisfies the conclusion of \cref{prop:optimization}.
\end{itemize} 

We will prove \cref{prop:opt-actual} using the sublinear primal-dual framework of Clarkson, Hazan, and Woodruff \cite[Algorithm~1, Algorithm~3]{CHW12}. The pseudocode is presented in \cref{algo:opt-actual} and relies on a few subroutines, which we now discuss. 

First, we need a standard iterative low-regret algorithm for the class of `experts' corresponding to the rescaled continuous cube $\mc{C} = [-1/\sqrt{n}, 1/\sqrt{n}]^{n}$. 

\begin{lemma}\label{lem:low-regret}
Consider a sequence of vectors $v_\ell\in \mb{R}^n$ with $\snorm{v_\ell}_2\le 1$ for $\ell\in [T]$ and let $\mc{C} = [-1/\sqrt{n}, 1/\sqrt{n}]^{n}$. The sequence of vectors $x_0 = 0$ and for $i \geq 1$,
\begin{align*}
    x_{i+1} = \on{LRA}(v_{i}, x_{i-1}) := \arg\min_{x\in \mc{C}}\|x - (x_{i-1} - \eta v_{i})\|_{2}^{2},
\end{align*}
where $\eta = \sqrt{2/T}$, satisfies
\begin{align*}
\sup_{\ell\in [T]}\frac{1}{T}\bigg(\max_{x \in \mc{C}}\sum_{i=1}^{\ell}v_i^{T}x - \sum_{i=1}^{\ell}v_{i}^{T}x_i\bigg)\le \sqrt{\frac{2}{T}}.
\end{align*}
Moreover, given $v_i$ and $x_{i-1}$, $\on{LRA}(v_{i}, x_{i-1})$ can be computed in time $O(n)$.  
\end{lemma}
\begin{proof}
The expression for $\on{LRA}(\cdot, \cdot)$ corresponds exactly to online mirror descent on $\mc{C}$ equipped with the $\ell_{2}$-norm, with respect to the $1$-strongly convex
mirror map $\Phi(x) = \frac{\|x\|_{2}^{2}}{2}$. 
Accordingly, we have the standard guarantee (see, e.g., \cite[Equation~4.10]{Bub15}) that for any $x \in \mc{C}$,
\begin{align*}
    \frac{1}{T}\sum_{i=1}^{\ell}v_{i}^{T}(x - x_i) &\leq \frac{1}{T}\bigg(\sup_{z \in \mc{C}}\frac{\|z\|_{2}^{2}}{\eta} + \frac{\eta}{2}\sum_{i=1}^{\ell}\|v_i\|_{2}^{2}\bigg) \leq \frac{1}{T\eta} + \frac{\eta}{2} = \sqrt{\frac{2}{T}},
\end{align*}
taking the balancing value $\eta = \sqrt{2/T}$. For the assertion about the running time, note that $(x_{i-1} - \eta v_{i})$ can be readily computed in time $O(n)$, and the minimization to compute $x_{i+1}$ can also be performed in time $O(n)$, since the closest point in $\mc{C}$ to a given point in $\mb{R}^{n}$ can be found in a coordinate-by-coordinate manner. 
\end{proof}

Next, for any $v \in \mb{R}^{n}$ with $\|v\|_{2} \leq 1$ and $x \in \mc{C}$, we need a fast, low-variance, unbiased estimator for $v^{T}x$. This is the content of the following lemma. 

\begin{lemma}\label{clm:estimate}
Let $x, v\in \mb{R}^{n}$ with $\snorm{x}_2,\snorm{v}_2\le 1$. Define 
$\on{Est}(x,v)$ to be $v_i/x_i$ with probability equal to $x_i^2$ and $0$ with probability $1-\snorm{x}_2^2$. Then, we have that 
\[\mb{E}[\on{Est}(x,v)] = x^Tv \quad \text{and} \quad \on{Var}[\on{Est}(x,v)] \le 1.\]
\end{lemma}
\begin{proof}
Note that 
\[\mb{E}[\on{Est}(x,v)] = \sum_{i\in [n]}v_i/x_i \cdot x_i^2 = v^Tx\]
and 
\[\on{Var}[\on{Est}(x,v)] \le \mb{E}[\on{Est}(x,v)^2] = \sum_{i\in [n]}v_i^2/x_i^2 \cdot x_i^2 \leq 1. \qedhere\]
\end{proof}

Finally, we require a data-structure which supports efficiently updating and sampling from probability distributions. We use (a much simpler version of) a data-structure provided in work of Carmon, Jin, Sidford, and Tian \cite[Section~2.4.1]{CJST20}. 

\begin{lemma}\label{lem:eff-sample}
There exists a data-structure for handling probability probability distributions on $[n]$ with the follow properties:
\begin{itemize}
    \item $\on{Initialization}({v})$: Given ${v}\in (\mb{R}_{\ge 0})^{n}$, construct the data structure corresponding to the probability distribution $v/\|v\|_{1}$ on $[n]$ in time $O(n)$. 
    \item $\on{Mult}(v, i, \tau)$: Given the data structure corresponding to the probability distribution determined by $v \in (\mb{R}_{\geq 0})^{n}$, a coordinate $i \in [n]$, and $\tau \in \mb{R}_{\geq 0}$, update to the data structure corresponding to the probability distribution determined by $v' \in (\mb{R}_{\geq 0})^{n}$ in time $O(\log{n})$, where $v'_i = \tau v_i$ and $v'_j = v_j$ for $j\neq i$. 
    \item $\on{Sample}({v})$: Given the data structure corresponding to the probability distribution determined by $v \in (\mb{R}_{\geq 0})^{n}$, produce a sample according to it in time $O(\log n)$. 
\end{itemize}
\end{lemma}
\begin{proof}[Proof sketch]
The data structure associated to $v \in (\mb{R}_{\geq 0})^{n}$ consists of an array on $[n]$, storing the entries of $v$, with a full binary tree on top. Each node in the tree maintains the sum of all the elements in the array which are its descendants (in particular, the sum at the root is $\|v\|_{1}$). $\on{Initialize}(v)$ constructs this tree in a `bottom-to-top' fashion and takes time $O(n)$ since there are $O(n)$ edges in this tree; $\on{Mult}(v,i,\tau)$ is implemented by starting at the $i^{th}$ position and `walking-up' to the root along the unique root-to-leaf path, updating the weights of the  $O(\log{n})$ nodes encountered on the way; $\on{Sample}(v)$ is implemented in $O(\log n)$-time by `walking down' from the root to a leaf, using the values at the left and right child of each node in order to toss a coin with suitable bias and decide whether to descend to the left child or to the right child. 
\end{proof}

The pseudocode for \cref{prop:opt-actual} is given in \cref{algo:opt-actual}. For a given vector $x\in \mb{R}^{n}$ with $\snorm{x}_2\le 1$, we define $\on{Dist}(x,\ell_2)$ to be the distribution $[n]$ specified by the square of the coordinates (outputting $\emptyset$ with probability $1-\snorm{x}_2^2$) and for $z,C \in \mb{R}$, define $\on{clip}(z,C) = \min\{\max\{z,-C\},C\}$. From the above description of various subroutines, it is immediate that $\textsc{Optimize}$ runs in the required time.

\begin{algorithm}[!ht]
\caption{Pseudocode for $\textsc{Optimize}(v',v_1,\dots, v_m, \epsilon)$ in \cref{prop:opt-actual} \label{algo:opt-actual}}

$T \assign (\log m)\epsilon^{-2}$\\
$x_0 \assign {0}$\\
$w_0 \assign \on{Initialize}(\mbm{1}_{m})$\\
$\eta \assign \sqrt{(\log m)/T}/100$\\
\For{$t = 1, \dots, T$}{
    $\tau_t \assign \on{Sample}(\on{Dist}(x_{t-1},\ell_2)))$\\
    $s_t\assign \on{Sample}(w_{t-1})$\\
    \uIf{$\tau_t\neq \emptyset$}{
        $v_{t}^* \assign \on{Clip}(v'^{T}e_{\tau_t}/x_{\tau_t}, 1/\eta)$ \\
        $J_{t} \assign \{j \in [m]: v_{j}^{T} e_{\tau_t} \neq 0\}$\\
        \For{$j \in J_{t}$}{
            $v_t(j)\assign \on{Clip}((v' + v_j)^{T}e_{\tau_t}/x_{\tau_t},1/\eta)$\\
            $w_t(j)\assign \on{Mult}(w_{t-1}, j, (1-\eta v_t(j) +\eta^2v_t(j)^2)\cdot(1-\eta v_t^* + \eta^2 (v_t^*)^{2})^{-1})$
        }
    }
    $x_t \assign \on{LRA}(x_{t-1}, v' + v_{s_t})$\\
}
\Return $\frac{1}{T}\sum_{i=1}^{T}x_i$\\
\end{algorithm}

\begin{lemma}\label{clm:runtime}
The runtime for $\textsc{Optimize}(v',v_1,\dots, v_m, \epsilon)$ is bounded by $\tilde{O}((k + n)\epsilon^{-2} + m)$.
\end{lemma}
\begin{proof}
Lines 1-4 take time $O(n+m)$. Line 15 takes time $O(Tn) = \tilde{O}(\epsilon^{-2}n)$. By \cref{lem:eff-sample}, each iteration of Lines 6-7 can be implemented in time $\tilde{O}(n)$, for a total runtime of $O(nT) = \tilde{O}(\epsilon^{-2}n)$. By \cref{lem:low-regret}, each iteration of Line 14 takes time $O(n)$, for a total runtime of $O(nT) = \tilde{O}(\epsilon^{-2}n)$. Each iteration of Line 9 takes time $O(1)$ and each iteration of Line 10 takes time $\tilde{O}(k)$ (see the footnote in the statement of \cref{prop:opt-actual}), so that together, all iterations of Lines 9-10 take time $\tilde{O}(Tk) = \tilde{O}(k\epsilon^{-2})$. Finally, by \cref{lem:eff-sample}, each iteration of Lines 11-13 takes time $\tilde{O}(k)$, for a total runtime of $\tilde{O}(kT) = \tilde{O}(k\epsilon^{-2})$.  \qedhere 

\end{proof}

Finally, we analyse the correctness of \cref{algo:opt-actual}. 

\begin{proof}[{Proof of \cref{prop:opt-actual}}]
The running time of $\textsc{Optimize}$ is analyzed in \cref{clm:runtime}. For the correctness, it is helpful to note that the probability distribution defined by the vector $w_{t} \in (\mb{R}_{\geq 0})^{m}$ in Line 13 coincides with the probability distribution defined by the vector $w'_{t} \in (\mb{R}_{\geq 0})^{m}$, where for $j \in [m]$,
\[w'_t(j) := w_{t-1}(j)\cdot (1-\eta v_t(j) + \eta^2 v_t(j)^2),\]
for $v_t(j)$ defined by the same formula as Line 12 (but now, for all $j \in [m]$). Indeed, $w_{t} = (1-\eta v_t^* + \eta^2 (v_t^*)^2)^{-1} \cdot w'_t$. 

With this observation, the conclusion follows essentially immediately from \cite[Algorithm~3]{CHW12}, upon noting that
\begin{itemize}
    \item $T_{\epsilon}(\on{LRA})$ is any $T$ for which the right hand side of \cref{lem:low-regret} is bounded above by $\epsilon$; clearly, $T\geq 4/\epsilon^{2}$ suffices, and
    \item by \cref{clm:estimate} and since $\snorm{v'+v_j}_2 \leq 1$ for all $j \in [m]$, we have that  $(v'+v_j)^{T}e_{\tau_t}/x_{\tau_t}$ is an unbiased estimator for $(v'+v_j)^T x$ with variance bounded by $1$. 
\end{itemize}

The only difference between \cite[Algorithm~3]{CHW12} and \cref{algo:opt-actual} is that in \cref{algo:opt-actual}, the estimators $(v+v_j)^{T}e_{\tau_t}/x_{\tau_t}$ are not independent for different $j\in [m]$. This only affects (potentially) the proofs of \cite[Lemma~B.3, Lemma~B.6, Lemma~B.7]{CHW12} but a trivial examination of the proof reveals that independence between the estimators is not used \footnote{This observation is already present in \cite[Algorithm~1]{CHW12}. In fact, our setting is essentially identical to \cite[Algorithm~1]{CHW12}, except that there, the minimization is over $x$ with $\snorm{x}_2\le 1$. Except for the use of the data structure in Line 13, our algorithm is identical to \cite[Algorithm~1]{CHW12} modulo noting that $[-1/\sqrt{n},1/\sqrt{n}]^{n}$ is contained in the unit $\ell_2$-ball and that the associated projection in mirror descent can be implemented efficiently.}. \qedhere
\end{proof}
\begin{remark}[{Numerical Precision Issues}]
The above analysis, as written, assumes exact arithmetic; there are two points which are numerically sensitive which can be handled using standard techniques. The first is dividing by $x_t$ in Lines 9 and 11. By \cite[Lemma~C.2]{CHW12}, it suffices to truncate entries smaller than $\on{poly}(n^{-1},\epsilon)$. 

The second point is keeping track of the vector $w_{t}$, and the induced probability distribution, in a numerically stable manner. This is discussed \cite[Section~G.1]{CJST20}; for our simplified data structure, however, a substantially simpler solution suffices, which we now sketch. 

First, note that the maximum and minimum value of any $w_t(i)$ over the course of the algorithm is bounded between $4^{-T}$ and $4^{T}$. We will maintain the logarithm of each $w_{t}(i)$ using $L = C(\log n + \log(1/\epsilon))$ bits, for a sufficiently large constant $C$. In particular, the version of $w_{t}(i)$ we work with is within a factor of $(1\pm \epsilon^{C}n^{-C})^{\tilde{O}((k+n)\epsilon^{-2} + m)}$ of the true $w_t(i)$, which is negligible for $C$ sufficiently large.

For maintaining the logarithm of the weights up to this precision in Line 13, note that, when `walking up' the binary tree, it suffices to set the logarithm of the value of a parent node to be equal to the logarithm of the value of the heavier child node, in the case when the logarithms of the values of the children differ by more than $2L$. Otherwise, denoting the value of the lighter child by $z$ and the heavier child by $y$, the logarithm of the value at the parent node is $\log(y+z) = \log(y) + \log(1+z/y)$, which can be computed in $O(\on{poly}(\log n,\log(1/\epsilon)))$-time to $L$ digits of precision, since $2^{-2L}\leq |z/y| \leq 1$.
Finally, when `walking down' the binary tree in Line 7, if the logarithms of the values of the children of a node differ by more than $L$, then it suffices to simply descend to the heavier child. Otherwise, the bias of the coin to flip is $y/(y+z) = 1/(1+z/y)$, which can be computed to the desired accuracy in time $O(\on{poly}(\log n,\log(1/\epsilon)))$, noting again that $2^{-L}\leq |z/y| \leq 1$.  
\end{remark}

\bibliographystyle{amsplain0.bst}
\bibliography{main.bib}

\providecommand{\bysame}{\leavevmode\hbox to3em{\hrulefill}\thinspace}
\providecommand{\MR}{\relax\ifhmode\unskip\space\fi MR }
\providecommand{\MRhref}[2]{%
  \href{http://www.ams.org/mathscinet-getitem?mr=#1}{#2}
}
\providecommand{\href}[2]{#2}
\begin{thebibliography}{10}

\bibitem{AO19}
Zeyuan Allen-Zhu and Lorenzo Orecchia, \emph{Nearly linear-time packing and
  covering {LP} solvers}, Mathematical Programming \textbf{175} (2019),
  307--353.

\bibitem{AW21}
Josh Alman and Virginia Vassilevska~Williams, \emph{A refined laser method and
  faster matrix multiplication}, Proceedings of the 2021 {ACM}-{SIAM}
  {S}ymposium on {D}iscrete {A}lgorithms ({SODA}), [Society for Industrial and
  Applied Mathematics (SIAM)], Philadelphia, PA, 2021, pp.~522--539.

\bibitem{AS16}
Noga Alon and Joel~H. Spencer, \emph{The probabilistic method}, fourth ed.,
  Wiley Series in Discrete Mathematics and Optimization, John Wiley \& Sons,
  Inc., Hoboken, NJ, 2016.

\bibitem{ALS21}
Ryan Alweiss, Yang~P. Liu, and Mehtaab Sawhney, \emph{Discrepancy minimization
  via a self-balancing walk}, S{TOC} '21---{P}roceedings of the 53rd {A}nnual
  {ACM} {SIGACT} {S}ymposium on {T}heory of {C}omputing, ACM, New York, [2021]
  \copyright 2021, pp.~14--20.

\bibitem{AHK12}
Sanjeev Arora, Elad Hazan, and Satyen Kale, \emph{The multiplicative weights
  update method: a meta-algorithm and applications}, Theory of computing
  \textbf{8} (2012), 121--164.

\bibitem{Ban98}
Wojciech Banaszczyk, \emph{Balancing vectors and {G}aussian measures of
  n-dimensional convex bodies}, Random Structures \& Algorithms \textbf{12}
  (1998), 351--360.

\bibitem{Ban10}
Nikhil Bansal, \emph{Constructive algorithms for discrepancy minimization},
  2010 {IEEE} 51st {A}nnual {S}ymposium on {F}oundations of {C}omputer
  {S}cience---{FOCS} 2010, IEEE Computer Soc., Los Alamitos, CA, 2010,
  pp.~3--10.

\bibitem{Bec81}
J{\'o}zsef Beck, \emph{Roth’s estimate of the discrepancy of integer
  sequences is nearly sharp}, Combinatorica \textbf{1} (1981), 319--325.

\bibitem{BF81}
J{\'o}zsef Beck and Tibor Fiala, \emph{“integer-making” theorems}, Discrete
  Applied Mathematics \textbf{3} (1981), 1--8.

\bibitem{Bub15}
S{\'e}bastien Bubeck, \emph{Convex optimization: Algorithms and complexity},
  Foundations and Trends{\textregistered} in Machine Learning \textbf{8}
  (2015), 231--357.

\bibitem{CJST20}
Yair Carmon, Yujia Jin, Aaron Sidford, and Kevin Tian, \emph{Coordinate methods
  for matrix games}, 2020 {IEEE} 61st {A}nnual {S}ymposium on {F}oundations of
  {C}omputer {S}cience, IEEE Computer Soc., Los Alamitos, CA, [2020] \copyright
  2020, pp.~283--293.

\bibitem{CHW12}
Kenneth~L. Clarkson, Elad Hazan, and David~P. Woodruff, \emph{Sublinear
  optimization for machine learning}, J. ACM \textbf{59} (2012), Art. 23, 49.

\bibitem{DadOpen}
Daniel Dadush,
  \url{https://homepages.cwi.nl/~dadush/workshop/discrepancy-ip/open-problems.html}.

\bibitem{ES18}
Ronen Eldan and Mohit Singh, \emph{Efficient algorithms for discrepancy
  minimization in convex sets}, Random Structures Algorithms \textbf{53}
  (2018), 289--307.

\bibitem{Gia97}
Apostolos~A. Giannopoulos, \emph{On some vector balancing problems}, Studia
  Math. \textbf{122} (1997), 225--234.

\bibitem{Glu88}
E.~D. Gluskin, \emph{Extremal properties of orthogonal parallelepipeds and
  their applications to the geometry of {B}anach spaces}, Mat. Sb. (N.S.)
  \textbf{136(178)} (1988), 85--96.

\bibitem{JSWZ20}
Shunhua Jiang, Zhao Song, Omri Weinstein, and Hengjie Zhang, \emph{Faster
  dynamic matrix inverse for faster lps}.

\bibitem{LRR17}
Avi Levy, Harishchandra Ramadas, and Thomas Rothvoss, \emph{Deterministic
  discrepancy minimization via the multiplicative weight update method},
  Integer programming and combinatorial optimization, Lecture Notes in Comput.
  Sci., vol. 10328, Springer, Cham, 2017, pp.~380--391.

\bibitem{LSS22}
Yang~P. Liu, Ashwin Sah, and Mehtaab Sawhney, \emph{{A Gaussian Fixed Point
  Random Walk}}, 13th Innovations in Theoretical Computer Science Conference
  (ITCS 2022) (Dagstuhl, Germany) (Mark Braverman, ed.), Leibniz International
  Proceedings in Informatics (LIPIcs), vol. 215, Schloss Dagstuhl --
  Leibniz-Zentrum f{\"u}r Informatik, 2022, pp.~101:1--101:10.

\bibitem{LM15}
Shachar Lovett and Raghu Meka, \emph{Constructive discrepancy minimization by
  walking on the edges}, SIAM J. Comput. \textbf{44} (2015), 1573--1582.

\bibitem{Rot17}
Thomas Rothvoss, \emph{Constructive discrepancy minimization for convex sets},
  SIAM J. Comput. \textbf{46} (2017), 224--234.

\bibitem{Spe85}
Joel Spencer, \emph{Six standard deviations suffice}, Trans. Amer. Math. Soc.
  \textbf{289} (1985), 679--706.

\end{thebibliography}

\appendix

\section{Deferred proofs from \texorpdfstring{\cref{sec:prelims}}{Section 2}}
\label{app:deferred}

\subsection{Proof of \texorpdfstring{\cref{thm:main}}{Theorem 1.1} given \texorpdfstring{\cref{thm:full-coloring}}{Theorem 2.1}}

The following result is an immediate consequence of \cite[Theorem~1.2]{ALS21}.

\begin{theorem}\label{thm:sparse-coloring}
There is a randomized algorithm $\on{Sparse-Coloring}$ and an absolute constant $C_{\ref{thm:sparse-coloring}}>0$ for which the following holds. On input a matrix $A \in \mb{R}^{m\times n}$, $\on{Sparse-Coloring}(A)$ runs in time $O(\on{nnz}(A) + n)$, and with probability at least $99/100$, returns a vector ${v}\in \{\pm 1\}^{n}$ such that 
\[\snorm{A{v}}_{\infty}\le C_{\ref{thm:sparse-coloring}} \snorm{A}_{1\to 2}\sqrt{(\log m)(\log n)}.\]
\end{theorem}

\begin{proof}[{Proof of \cref{thm:main}}]
First, note that we may assume that $m\le n^{2}$; otherwise \cref{thm:main} follows from noting (using Bernstein's inequality and the union bound) that a uniformly random $v \in \{\pm 1\}^{n}$ succeeds with probability at least $1/2$ (and given $v$, its success can be checked in time $O(\on{nnz}(A) + n)$). Moreover, we may assume that $m\ge n/(\log n)^2$; otherwise, $\snorm{A}_{1\to 2}\le \sqrt{m}\cdot \|A\|_{1 \to \infty} \leq \sqrt{n}/(\log{n})$, and we may use the algorithm $\on{Sparse-Coloring}$ from \cref{thm:sparse-coloring}. 

Now, suppose that $n/(\log n)^2\le m\le n^2$. Define the sets 
\[C_{\on{Light}} := \{i\in [n]: \snorm{Ae_i}_2\le \sqrt{n}/(\log n)\} \quad \text{and} \quad C_{\on{Heavy}} := [n]\setminus C_{\on{Light}},\]
where $e_i$ denote the $i$-th elementary basis vector. Note that for any $j \in C_{\on{Heavy}}$, we have
\[\sqrt{\on{nnz}(Ae_j)}\|A\|_{1 \to \infty} \geq \sqrt{n}/(\log{n}),\]
from which we see that
\[\on{nnz}(Ae_j) = \tilde{\Omega}(n).\]

We define $A_{\on{Light}}$ to be the restriction of $A$ to columns corresponding to $\mb{R}^{C_{\on{Light}}}$ and $A_{\on{Heavy}}$ to be the restriction of $A$ to columns corresponding to $\mb{R}^{C_{\on{Heavy}}}$. Note that it suffices to find a vector $v\in \{\pm 1\}^{C_{\on{Heavy}}}$ such that $\snorm{A_{\on{Heavy}}v}_{\infty}\lesssim \sqrt{n\log(m/n + 2)}$ in time $\tilde{O}(\on{nnz}(A) + n)$ as $A_{\on{Light}}$ is handled immediately by \cref{thm:sparse-coloring}. By \cref{thm:full-coloring}, we can find such a $v$ in time $\tilde{O}(\on{nnz}(A_{\on{Heavy}}) + |C_{\on{Heavy}}|^{2}) = \tilde{O}(\on{nnz}(A) + |C_{\on{Heavy}}|^{2})$, which we claim is $\tilde{O}(\on{nnz}(A) + n)$. Indeed, 
\[\on{nnz}(A_{\on{Heavy}}) \geq |C_{\on{Heavy}}|\cdot \min_{j \in C_{\on{Heavy}}}\on{nnz}(Ae_j) = \tilde\Omega(|C_{\on{Heavy}}|\cdot n)\]
so that
\[|C_{\on{Heavy}}|^{2} \leq |C_{\on{Heavy}}|\cdot n = \tilde{O}(\on{nnz}(A_{\on{Heavy}})) = \tilde{O}(\on{nnz}(A)). \qedhere\]
\end{proof}

\subsection{Proof of \texorpdfstring{\cref{thm:full-coloring}}{Theorem 2.1} given \texorpdfstring{\cref{thm:partial-coloring}}{Theorem 2.2}}

\begin{proof}[{Proof of \cref{thm:full-coloring}}]

As in the previous subsection, it suffices to consider the case $n/(\log n)^2\le m\le n^2$. 

Initialize ${v}_0 = {0}$. At each time step $\ell \geq 1$, given the partial coloring ${v}_{\ell-1} \in [-1,1]^{n}$, let $\mc{F}_{\ell} = \{i \in [n]: ({v}_{\ell-1})_i\in \{\pm 1\}\}$, $\mc{G}_{\ell} = \{i \in [n]: ({v}_{\ell-1})_i\notin \{\pm 1\}\}$, and $\Lambda_{\ell} = \on{Diag}({1}-|{v}_{\ell-1}|)$. If $\mc{G}_{\ell} = \emptyset$, then return $v_{\ell-1}$. Else, let $A_{\ell}$ denote the restriction of $A$ to the columns spanned by $\mc{G}_{\ell}$ and notice that $\Lambda_{\ell}$ restricts naturally to $\mc{G}_{\ell}$. 

We consider two separate cases. If $|\mc{G}_{\ell}| \geq n/(\log{n})^{2}$, then we use  \cref{thm:partial-coloring} to find a vector ${v_{\ell}'}\in [-1,1]^{\mc{G}_{\ell}}$ such that 
$\snorm{A_{\ell}\Lambda_{\ell}{v}'_{\ell}}_{\infty}\le C_{\ref{thm:partial-coloring}}\sqrt{|\mc{G}_{\ell}|\log(m/|\mc{G}_{\ell}| + 2)}$ and such that ${v}'_{\ell}$ has at least $C_{\ref{thm:partial-coloring}}^{-1}|\mc{G}_{\ell}|$ coordinates that are valued in $\{\pm 1\}$ \footnote{More precisely, we make at most $\tilde{O}(1)$ independent calls to \cref{thm:partial-coloring}, which guarantees that we find such a $v'_{\ell}$ with probability at least $1-1/n$.}. In particular, for at least one value of $\sigma_{\ell} \in \{\pm 1\}$, the vector $v_{\ell} := {v}_{\ell-1} + \sigma_{\ell} \Lambda_{\ell}{v}'_{\ell}$ has at most $(1-C_{\ref{thm:partial-coloring}}^{-1}/2)|\mc{G}_{\ell}|$ coordinates that are not valued in $\{\pm 1\}$.
Moreover,
\begin{align*}
    \|Av_{\ell}\|_{\infty} \leq \|Av_{\ell-1}\|_{\infty} + \|A_{\ell}\Lambda_{\ell}v'_{\ell}\|_{\infty} \leq \|Av_{\ell-1}\|_{\infty} + C_{\ref{thm:partial-coloring}}\sqrt{\mc{G}_{\ell} \log(m/|\mc{G}_{\ell}| + 2)}.
\end{align*}

On the other hand, if $|\mc{G}_{\ell}|\le n/(\log n)^2$, let $v'_{\ell} \in [-1,1]^{n}$ denote the random vector with independent coordinates such that $(v'_{\ell} + v_{\ell-1})_i \in \{\pm 1\}$ and $\mb{E}[(v'_{\ell})_i] = 0$ for all $i \in [n]$. Let $v_{\ell} = v'_{\ell} + v_{\ell-1} \in \{\pm 1\}^{n}$. A direct application of Bernstein's inequality and the union bound shows that, with probability at least $99/100$,
\[\|Av_{\ell}\|_{\infty} \leq \|Av_{\ell-1}\|_{\infty} + \|Av'_{\ell}\|_{\infty} \leq \|Av_{\ell-1}\|_{\infty} + 10\sqrt{n}.\]
Note that, given $v_{\ell-1}$, we can sample from $v'_{\ell}$ in time $O(n)$ and verify that $v_{\ell}$ satisfies the above inequality in time $O(\on{nnz}(A) + n)$. 

Observe that, due to the guarantee $|\mc{G}_{\ell+1}| \leq (1-C_{\ref{thm:partial-coloring}}^{-1}/2)|\mc{G}_{\ell}|$, we are in the first case for at most $O(\log\log{n})$ iterations, which together take time
\[\tilde{O}\left(\sum_{\ell=0}^{O(\log\log{n})}(1-C_{\ref{thm:partial-coloring}}^{-1}/2)^{2\ell-2}n^2 + \on{nnz}(A_{\ell})\right) = \tilde{O}(n^{2} + \on{nnz}(A)).\]
As mentioned before, the second step takes time $O(\on{nnz}(A)+n)$. 

Finally, denoting the output of the process by $v \in \{\pm 1\}^{n}$ and using $|\mc{G}_{\ell}| \lesssim (1-C_{\ref{thm:partial-coloring}}^{-1}/2)^{\ell}n$, we have that 

\begin{align*}
\|Av\|_{\infty} &\lesssim \sqrt{n} + \sum_{\ell\le O(\log\log{n})}\sqrt{|\mc{G}_{\ell}|\log(m/|\mc{G}_{\ell}|+2)} \lesssim \sqrt{n\log(m/n+2)},
\end{align*}
as desired. 
\end{proof}

\end{document}